\documentclass[12pt]{article}
\usepackage{amsmath,amsthm, enumitem, dcolumn, graphicx, tikz,pdfpages, caption, float}

\usepackage[ backend=bibtex, style=authoryear-comp, isbn=false,url=false,doi=false,eprint=false,sorting=nyt]{biblatex}
\DeclareSourcemap{
	\maps[datatype=bibtex]{
		\map[overwrite]{
			\step[fieldsource=doi, final]
			\step[fieldset=url, null]
			\step[fieldset=eprint, null]
			\step[fieldset=issn, null]
			\step[fieldset=ISSN, null]
			\step[fieldset=isbn, null]
		}  
	}
}
\DeclareUnicodeCharacter{2009}{ }
\usetikzlibrary{trees}
\DeclareMathOperator*{\plim}{plim}
\newtheorem{theorem}{Proposition}
\def\eolqed{\hspace{\stretch1}\ensuremath\qedsymbol}
\DeclareUnicodeCharacter{2606}{\ensuremath{\star}}

\title{Estimating Demand for a New Product}
\author{Sizhong Sun \\ James Cook University \\ Townsville, Qld, Australia \\ Email: sizhong.sun@jcu.edu.au
	\thanks{The author would like to thank seminar participants of College of Business Law and Governance, James Cook University, and University of Tasmania for the comments and suggestions. All errors remain the author's.}}

\begin{document}
	\maketitle
	
	\begin{abstract}
		This paper develops an approach for estimating demand for a new product. Taking willingness to pay (WTP) as primitive, it establishes a general and yet analytically simple demand function, and proposes an estimation procedure that consistently recovers the underlying demand function from the WTP data. Monte Carlo simulations find the estimation procedure works well in identifying the demand function. This approach complements existing methods of demand estimation, and can be applied both within and outside academia, for example in teaching economics, for a business to launch new products, and for policymakers to conduct non-market valuation.   
		
		
		
		\pagebreak
	\end{abstract}
	
	\section{Introduction}
	As a fundamental concept, there are few fields in economics where the demand function does not play an essential role. In addition, demand is also important outside academia. For example, a business will be interested to understand the demand for its products, and a government agency will be keen to know the public's demand for its service. In light of its importance, researchers have also long been searching for different methods of demand estimation, resulting in a large volume of existing literature tackling different aspects of demand estimation.\footnote{See for examples \textcite{10.2307/2233510}, \textcite{doi:10.1111/b.9780631216339.1999.00005.x}, \textcite{BARNETT2008210}, \textcite{BERRY20211}, and \textcite{GANDHI202163} for surveys.}
	
	Demand estimation frequently starts with a particular form of the demand function, derived from consumers' utility maximization, and use observational data to estimate parameters of interest. The Cobb-Douglas (logarithmic linear) demand function and the constant elasticity of substitution (CES) demand function are widely used in practice, possibly due to the fact that they globally satisfy the regularity conditions required for utility maximization. Other functional forms include the generalized Leontief function \parencite{DIEWERT1973284}, the translog functional form \parencite{10.2307/1804840}, the almost ideal demand system (AIDS) \parencite{10.2307/1805222}, the minflex Laurent model \parencite{BARNETT198331, BARNETT198533, 10.2307/1913216, BARNETT19853, BARNETT1987281, doi:10.1080/07350015.1994.10524565}, the quadratic AIDS \parencite{doi:10.1162/003465397557015}, the normalized quadratic reciprocal indirect utility function \parencite{doi:10.1080/07350015.1988.10509667}, the normalized quadratic expenditure function \parencite{doi:10.1080/07350015.1988.10509667}.\footnote{More details of these functions and a survey of non-parametric and semi-parametric approaches of demand analysis can be found in \textcite{BARNETT2008210}. A comparison of the performance of these functions can be found in \textcite{doi:10.1002/jae.585} and \textcite{HOLT2002117}.}
	
	Researchers have also estimated demand as a choice modeling. For example, to estimate demand for differentiated products, consumer behavior is modeled as choosing a product from a set of available options to maximize their utility, resulting in a market share as a function of a set of explanatory variables. Among others, \textcite{10.2307/2171802} provide a framework for the choice modeling approach of demand analysis. \textcite{doi:10.3982/ECTA9027} discuss the identification results using market level data, and \textcite{https://doi.org/10.3982/ECTA20731} focus on a nonparametric identification using micro data. \textcite{FREYBERGER2015162}, \textcite{10.2307/2171802}, and \textcite{10.1111/j.1467-937X.2004.00298.x} develop asymptotic theories for the estimation of differentiated product demand models with respect to the number of markets and products. \textcite{https://doi.org/10.3982/QE1593} address the issue of many zeroes in market share data in estimating differentiated product demand systems. It shall be noted that the choice modeling of consumers' behavior implies a willingness to pay (WTP) by consumers for a particular product. \textcite{doi:10.1002/jae.1187} shows a duality between market shares and consumer types.
	
	In demand estimations, one typically utilizes observational data. Such observational data are outcomes of market equilibrium, namely the observed prices are functions of the market supply and demand, resulting in an endogeneity issue in estimations. To address the endogeneity issue, researchers typically use the supply-side factors as excluded instruments. 
	
	Observational data are not available if the product studied is a new product. As such, to estimate the demand function for a new product, researchers typically rely on the survey data, where respondents are asked to report the quantity of demand, given a price and a set of features of the studied product \parencite [See for example] [] {KLEIN199753, CampbellNewproductdemand}.\footnote{A related strand of studies is to forecast the demand for a new product, for example \textcite{BROCKHOFF1993211}, \textcite{LEE20121280}, \textcite{VANSTEENBERGEN2020113401}, \textcite{Yamamura04072022}, and \textcite{SOUSA2025125313}.} 
	
	In this study, I propose an alternative approach for estimating the demand function of a new product. I take consumers' WTP for a goods as model primitive, and establish a demand function under three behavioral assumptions. With WTP data, I first non-parametrically estimate its cumulative distribution function (CDF) and the quantity of demand. Then the ordinary least squares (OLS) estimator can be applied to estimate the demand function. By construction, the OLS estimation produces consistent parameter estimate, as the right-hand-side variables (e.g. price) are asymptotically orthogonal to the error term. Monte Carlo simulations find that this approach works well in identifying the underlying parameters of interest.
	
	This research starts from consumers' WTP, and hence extends existing literature on measuring WTP. Methods to measure consumers' WTP can be broadly categorized from two dimensions, namely whether the context is real or hypothetical, and whether the measurement is direct or indirect. Measurements of actual WTP can be obtained from lotteries, auctions (for example the Vickrey auction), and other market transactions (revealed preference data). Auctions can reveal consumers' WTP directly, while using the revealed preference data researchers can uncover WTP indirectly. In the hypothetical context, researchers utilize the contingent valuation method (CVM) \parencite{doi.org/10.4324/9781315060569} to directly solicit WTP. Indirectly, the conjoint analysis \parencite{10.1086/208721} is a widely used method to determine WTP. Researchers have discussed the strength and weakness of these methods and extended them in different directions \parencite [See for example, among others,] [] {doi.org/10.1007/s11747-019-00666-6, doi:10.1509/jmkr.48.1.172, LOUVIERE201057, 10.1257/jep.26.4.27, doi:10.1287/mksc.1070.0321, doi:10.1007/s11002-006-5147-x}. 
	
	This study contributes to existing literature from two aspects. First, the proposed demand function is general and yet analytically simple, and can be readily used in other more sophisticated economic modelling. Under specific assumptions on the distribution of WTP, the demand function reduces to commonly used functional forms. For example, if WTP is uniformly distributed, the implied demand is linear, while with Pareto distributed WTP, the associated demand function is logarithmic linear. 
	
	Second, the estimation procedure I propose is easy to be implemented, and can be applied in many situations, both within and outside academic arenas. Within the academic arena, one can use this approach in teaching economics. In class, for example, a lecturer can quickly survey students' WTP for a cup of coffee, and use a spreadsheet program to intuitively illustrate and estimate the demand for coffee. By construction (see the next section), the estimated demand will satisfy the Law of Demand. In the public sector, a government agency can utilize this approach to estimate a demand for health service, which can then be used to evaluate a certain health policy.\footnote{For example, \textcite{SUN2022101135} estimates the demand for a COVID-19 vaccine in China and UAE, and compute the consumer surplus of vaccination.} In the private sector, a business can use this approach to estimate the demand for its new products, which shall help the business to achieve a better performance.
	
	The study is linked to existing literature in four dimensions. First, the choice modelling of demand implies consumers' WTP for a goods, and as such the study is linked to this strand of research, despite I take WTP as primitive. Second, this study is also linked to those demand studies that estimate different functional forms of demand, using observational data, in that the demand function I propose can imply these functional forms. Compared with estimating demand from observational data, this approach is not subject to the potential endogeneity issue that arises from the fact that the observational data are market equilibrium outcomes. Third, in particular, this study contributes to estimating the demand function for a new product. Without observational data available, one can, for example, implement a Vickery auction to solicit consumers' WTP, where truthful bidding a weakly dominant strategy. With WTP data, one can then proceed to estimate the demand function. Fourth, this study is related to the strand of research that measures WTP, in that it requires WTP data as input. A number of methods can be used to measure consumers' WTP, as discussed in the previous paragraph. 
	Later, as an example, I will use the contingent valuation method (CVM) to solicit WTP.	However, it shall be noted that this study is not about the methods of measuring WTP themselves. 
	
	The method this study proposes intends to complement the existing methods of demand estimation. If data are available, one can estimate the demand function using both methods in order to cross-check the estimations. Besides, since I focus on new products,\footnote{Though it can also be used to estimate the demand function of existing products, should WTP data be available.} the demand function is necessarily static, as consumers have not yet consumed the products. Therefore, this study is not related to the dynamic demand estimation,\footnote{For example see, among others, \textcite{DESCHAMPS1998335}, \textcite{KHWAJA2010130} and \textcite{doi:10.1111/1756-2171.12194}.} which is an area for future research.         
	
	The rest of the paper is organized into six sections. In Section two, I establish the theoretical demand function, which guides the estimator in Section three. Section four discusses the possible mis-reporting of WTP and the robustness of the estimator. Section five reports a number of Monte Carlo simulations that test the effectiveness of the estimator. In Section six, I estimate a demand for coffee with survey that solicits consumers' WTP. Section seven concludes the paper.
	
	\section{The Demand Function}
	In order to establish the demand function, I start with the following behavioral assumptions:

	\begin{itemize}[rightmargin=20pt, leftmargin=40pt, itemsep=0pt, parsep=0pt]
		\item [A1:] A consumer is allowed to purchase one unit of the goods;
		\item [A2:] Consumers have willingness to pay (WTP), denoted by $W$, for the goods. $W$ is randomly distributed with a conditional cumulative distribution function (CDF) $F_{W|Z}(w)$, where $Z$ is an index that captures all other factors, such as income and prices of substitute and complement goods, that affect a consumer's WTP\footnote{Alternatively, $Z$ can be a vector of all the other factors. In such case, the design of $Z$ is subject to the curse of dimensionality.}, and $W|Z$ represents random variable $W$ conditional on random variable $Z$;
		\item [A3:] A consumer will purchase the goods if its price is lower than her/his WTP. 
	\end{itemize}
	
	Assumption A1 is harmless. It holds in may situations. For example, within a short time frame, a consumer is willing to purchase only one cup of coffee. Besides, as will be shown later, one can easily relax this assumption. Assumption A3 can also be relaxed by assuming that a consumer has a probability of actually purchasing the goods if her/his WTP is higher than the price, where this probability does not depend on price. 
	
	In Assumption A2, I take WTP as the model primitive. However, one can show that there exists WTP for a goods, using the utility function as primitive. Let consumer \textit{i}'s preference be represented by a continuous and quasi-concave utility function $U(\chi, \textbf{q})$, where $\chi\in\{0,1\}$, taking a value of one if the consumer consumes the goods and zero otherwise, and $\textbf{q}$ denotes a bundle of all the other goods\footnote{Alternatively, $\textbf{q}$ can be a vector of all the other goods.}. Subject to her/his budget constraint, the consumer chooses a bundle of goods $(\chi,\textbf{q})$ to maximize her/his utility, namely $\max_{\{\chi,\textbf{q}\}} U(\chi, \textbf{q})$ s.t. $p\chi+P\textbf{q} \leq I, p>0, P>0$, where $p$ denotes price of the goods, $P$ is the price of the bundle of all the other goods, and $I$ represents income. 
	
	If the consumer decides not to purchase the goods, her/his utility maximization problem reduces to $\max_{\{\textbf{q}\}} U(0, \textbf{q})$ s.t. $P\textbf{q} \leq I, P>0$. Let $V_{i0}(P,I)$ be the associated maximized value, and by the Maximum Theorem, $V_{i0}$ is continuous with respect to $(P,I)$. Similarly, if the consumer decides to purchase the goods, her/his problem then becomes $\max_{\{\textbf{q}\}} U(1, \textbf{q})$ s.t. $p + P\textbf{q} \leq I, p>0, P>0$. The associated maximized value is denoted as $V_{i1}(p,P,I)$, which is a continuous function of $(p,P,I)$ by the Maximum Theorem. To decide whether to purchase the goods, the consumer compares $V_{i0}(P,I)$ with $V_{i1}(p,P,I)$. That is, $\chi =1(V_{i1}(p,P,I) > V_{i0}(P,I))$, where $1()$ is an indicator function. 
	
	Since the consumer values the goods, namely $U(1, \textbf{q}) > U(0, \textbf{q})$ for any $\textbf{q}$, $\lim_{p\rightarrow0} V_{i1}(p,P,I) > V_{i0}(P,I)$ for any $(P,I)$. Together with the fact that $V_{i1}$ is a decreasing function of $p$, there exists a $\overline{p}$ such that $\forall p>\overline{p}, \chi=0$. Hence, $\overline{p}$ can be defined as the consumer's WTP. Note that different consumers have different WTP, and it is assumed that WTP can be characterized by a distribution. 
	
	Let $L$ denote the number of consumers in the market. The behavioral assumptions, A1-A3, imply the following demand function if $Z$ is a discrete random variable:
	\begin{equation}
	q=q(p,z)=[1-F_{W|Z=z}(p)]\Pr(Z=z)L
	\end{equation}
	where $q$ is the quantity of demand, and $F$ denotes the CDF. Note that since a CDF is a non-decreasing function, Equation (1) satisfies the law of demand. If $Z$ is a continuous random variable that admits a probability density function $f_{Z}(z)$, let $Q=Q(p,z,h)=[1-F_{W|Z\in[z-h,z+h]}(p)][F_{Z}(z+h)-F_{Z}(z-h)]L$ which is the demand conditional on $Z\in[z-h,z+h]$. Then conditional on $Z=z$, the demand function can be written as follows:
	\begin{equation}
	\begin{split}
	q=q(p,z)&=\underset{h\rightarrow0}{\lim}\frac{Q(p,z,h)}{2h} \\
	&=\underset{h\rightarrow0}{\lim}\frac{[1-F_{W|Z\in[z-h,z+h]}(p)][F_{Z}(z+h)-F_{Z}(z-h)]L}{2h} \\
	&=[1-F_{W|Z}(p)]f_{Z}(z)L
	\end{split}
	\end{equation}
	where similar to Equation (1), the Law of Demand is satisfied by construction.
	
	In addition, Equations (1) and (2) are very general Marshallian demand functions in that with different distributions of WTP and $Z$, they will reduce to demand functions that are commonly used. For example, if $W|Z\sim U[0,\overline{w}]$ and $Z\sim U[0,\overline{z}]$, the implied demand function is linear, namely $q=q(p,z)=q(p)=\beta_{0}+\beta_{1}p$, where $\beta_{0}=\frac{L}{\overline{z}}$ and $\beta_{1}=-\frac{L}{\overline{w}\overline{z}}$. If $W|Z\sim Exp(\lambda)$ and $Z^\alpha\sim U[0,\overline{z}]$, the demand function becomes semi-logarithmic linear, namely $\ln(q)=\beta_{0}+\beta_{1}p+\beta_{2}\log(z)$, where $\beta_{0}=\log(\frac{\alpha L}{\overline{z}})$, $\beta_{0}=-\lambda$ and $\beta_{2}=\alpha - 1$. If $W|Z\sim Pareto(\underline{w},a)$ over the support $[\underline{w},+\infty]$ with a shape parameter of $a$ and $Z^\alpha\sim U[0,\overline{z}]$, the associated demand function reduces to a  logarithmic linear function,  $\ln(q)=\beta_{0}+\beta_{1}\log(p)+\beta_{2}\log(z)$, where $\beta_{0}=\log(\frac{\alpha \underline{w}^a L}{\overline{z}})$, $\beta_{0}=-a$ and $\beta_{2}=\alpha - 1$. 
	
	If the distributions of WTP and $Z$ are unknown, assuming the CDF of WTP and PDF of $Z$ are smooth, one can approximate the demand function by Taylor expansion at $(p_{0},z_{0})$, an expansion point chosen by researchers\footnote{Note one can also first take logarithm at both sides of Equations (1) and (2) and expand at $(\log(p_{0}),\log(z_{0}))$.}. A first-order Taylor expansion yields a linear demand function, as follows:
	\begin{equation}
	q=q(p,z)\approx\beta_{0}+\beta_{1}p+\beta_{2}z
	\end{equation}  
	where $\beta_{0}=q(p_{0},z_{0})-p_{0}\frac{\partial q(p_{0},z_{0})}{\partial p}-z_{0}\frac{\partial q(p_{0},z_{0})}{\partial z}$, $\beta_{1}=\frac{\partial q(p_{0},z_{0})}{\partial p}$, and $\beta_{2}=\frac{\partial q(p_{0},z_{0})}{\partial z}$. A second-order Taylor expansion results in the following quadratic demand function:
	\begin{equation}
	q=q(p,z)\approx\beta_{0}+\beta_{1}p+\beta_{2}p^2+\beta_{3}z+\beta_{4}z^2+\beta_{5}pz
	\end{equation} 
	where 
	$\beta_{0}=q(p_{0},z_{0})
	-p_{0}\frac{\partial q(p_{0},z_{0})}{\partial p}
	+p_{0}^2\frac{\partial ^2 q(p_{0},z_{0})}{\partial p^2}
	-z_{0}\frac{\partial q(p_{0},z_{0})}{\partial z}
	+z_{0}^2\frac{\partial ^2 q(p_{0},z_{0})}{\partial z^2}
	+p_{0}z_{0}\frac{\partial ^2 q(p_{0},z_{0})}{\partial p \partial z}
	$, 
	$\beta_{1}=\frac{\partial q(p_{0},z_{0})}{\partial p}
	-2p_{0}\frac{\partial ^2 q(p_{0},z_{0})}{\partial p^2}
	-z_{0}\frac{\partial ^2 q(p_{0},z_{0})}{\partial p \partial z}
	$,
	$
	beta_{2}=\frac{\partial ^2 q(p_{0},z_{0})}{\partial p^2}
	$, 
	$\beta_{3}=\frac{\partial q(p_{0},z_{0})}{\partial z}
	-2z_{0}\frac{\partial ^2 q(p_{0},z_{0})}{\partial z^2}
	-p_{0}\frac{\partial ^2 q(p_{0},z_{0})}{\partial p \partial z}
	$,
	$
	\beta_{4}=\frac{\partial ^2 q(p_{0},z_{0})}{\partial z^2}
	$, and 
	$\beta_{5}=\frac{\partial ^2 q(p_{0},z_{0})}{\partial p \partial z}
	$. In order to achieve a higher degree of accuracy, one can, not surprisingly, increase the order of expansion, provided that the CDF and PDF are sufficiently smooth.
	
	To relax the assumption A1, suppose a consumer is allowed to purchase $M$ units of goods, where $M$ is a random variable. If $M|Z$ is discretely distributed over the support $\{0,1,\ldots,\overline{M}\}$, the associated demand function can be derived as follows:
	\begin{align*}
	q&=q(p,z)=\sum_{m=0}^{\overline{M}}\Pr(M=m|Z)[1-F_{W|Z,M}(p)]f_{Z}(z)Lm \\
	&= [1-F_{W|Z}(p)]f_{Z}(z)L\sum_{m=0}^{\overline{M}}\Pr(M=m|Z)m
	\end{align*}
	where it is assumed that $W$ (WTP) is independent of $M$ in the third equality. Similarly, if $M|Z$ is a continuous random variable with a probability density function (PDF) $f_{M|Z}(m)$ over the support $[0,\overline{M}]$, the corresponding market demand function can be written as:
	\begin{align*}
	q&=q(p,z)=\int_{m=0}^{\overline{M}}f_{M|Z}(m)[1-F_{W|Z,M}(p)]f_{Z}(z)Lm dm \\
	&= [1-F_{W|Z}(p)]f_{Z}(z)L\int_{m=0}^{\overline{M}}f_{M|Z}(m)m dm
	\end{align*}
	Note $Pr(M=m|Z)$ and $f_{M|Z}(m)$ can be functions of $Z$. For example, as income increases, a consumer is willing to purchase more goods, \textit{ceteris paribus}.    
	
	\section{The Estimator}
	To empirically estimate the demand function (Equations 1 and 2), one can utilize an estimation procedure outlined in Section 3.1. In brief, one needs to obtain the WTP data, for example by using the Vickrey auction or CVM to solicit consumers' WTP. With the WTP data, one non-parametrically estimates the CDF of WTP and PDF of $Z$ (or probability mass function of $Z$ if it is a discrete random variable), which is then used to construct the quantity of demand. The ordinary least squares (OLS) estimator can then be used to estimate the parameters of interest. 
	
	The estimation procedure depends on consumers revealing their WTP truthfully. Some methods can achieve this aim better than the others. For example, in the Vickrey auction,truthfully bidding is a weakly dominant strategy. Nevertheless, in order to check the robustness of the estimation procedure, I discuss mis-reporting of WTP in Section 3.2. Under some assumptions, it can be shown that the estimation procedure is robust to small mis-reporting of WTP.
	
	\subsection{Estimation procedure}
	The estimation procedure consists of the following five steps:
	\begin{itemize}[rightmargin=30pt, leftmargin=50pt, itemsep=0pt, parsep=0pt]
		\item [Step 1:] Obtain data of WTP and $Z$, $\{(w_{i},z_{i}): i=1,\ldots,n\}$;
		\item [Step 2:] Conditional on $Z$, use the data from Step 1 to estimate $F_{W|Z}(w)$ as $\hat{F}_{W|Z}(w)=\frac{\sum_{i=1}^{n}1(w_{i}\leq w)}{n}$;
		\item [Step 3:] Non-parametrically estimate $\hat{f}_{Z}(z)$ (or $\Pr(Z=z)$ if $Z$ is a discrete random variable);
		\item [Step 4:] Estimate the quantity of demand: $\forall p\in\{\bar{w}_{1},\bar{w}_{2},\ldots,\bar{w}_{N}\}, \hat{q}=[1-\hat{F}_{W|Z}(p)]\hat{f}_{Z}(z)L$, where $L$ is the number of consumers in the market, a known parameter,\footnote{Note here I implicitly assume that a consumer, faced with a price at his/her WTP, will not purchase the goods. Alternatively one can assume that in such a situation the consumer will purchase the goods, in which case $\hat{F}_{W|Z}(w)=\frac{\sum_{i=1}^{n}1(w_{i} < w)}{n}$} and $\{\bar{w}_{1},\bar{w}_{2},\ldots,\bar{w}_{N}\}$ is the set of distinct values in $\{w_{1},w_{2},\ldots,w_{n}\}$ with $N \leq n$ and $\underset{n\rightarrow\infty}{\lim} \frac{N}{n} > 0$;
		\item [Step 5:] Choose $\beta$ to minimize the sum of squared errors between the quantities of demand and their estimates: $\beta=\arg\min\sum_{i=1}^{N}(q_{i}-\hat{q}_{i})^2$. 
	\end{itemize}
	
	Step 1 can be implemented using different methods, for example the conjoint analysis, auctions, and contingent valuation method\footnote{Using different methods frequently amounts to whether to use stated or revealed preference data. Among others, a comparison can be found in \textcite{https://doi.org/10.1002/hec.4246}.}. Later, as an example, I will use the CVM to solicit consumers' WTP. The CVM particularly suites commodities that do not have real markets, and has been extensively used in such situations. \textcite{doi.org/10.4324/9781315060569}, \textcite{doi:10.1021/es990728j}, \textcite{CARSON2005821}, and \textcite{10.1257/jep.26.4.27} provide detailed guidelines on conducting CVM. 
	
	In conducting CVM, the survey shall be designed to include/control $Z$. For example, if a consumer's income is likely to affect her/his WTP, the questionnaire shall ask the interviewees to report their income. Besides, as in any survey, one needs to carefully consider the non-sampling errors\footnote{See, for example, \textcite{Tanur2011} for a discussion.}. For example, if the survey starts by asking whether the interviewee is willing to pay a certain amount of money for the goods, there is likely an anchor point effect. That is, the interviewee's reported WTP will be correlated with the starting figure in the survey. To overcome this anchor point effect, one can randomize the starting figure in the survey.
	
	In Step 2, by the Glivenko-Cantelli Theorem, the empirical CDF converges to the underlying CDF, namely $\underset{n\rightarrow\infty}{\lim}\hat{F}_{W|Z}(w)=F_{W|Z}(w)$. Similarly, in Step 3, $\hat{f}_{Z}(z)$ converges to $f_{Z}(z)$. Accordingly, in Step 4, the estimated quantity of demand converges to the true quantity of demand, namely $\underset{n\rightarrow\infty}{\lim}\hat{q}=q$.
	
	Let $\textbf{X}$ denote the $N\times k$ matrix of explanatory variables, $\hat{\textbf{Y}}$ denote $N\times 1$ vector of dependent variables (the estimated quantity of demand), and $\textbf{Y}$ denote $N\times 1$ vector of the true (unobserved) quantity of demand, as follows:
	\begin{align*}
	\textbf{X} =
	\begin{pmatrix}
	p_{1}&z_{1}\\
	p_{2}&z_{2}\\
	\vdots  & \vdots\\
	p_{N}&z_{N}\\
	\end{pmatrix}    
	\hat{\textbf{Y}}\ =
	\begin{pmatrix}
	\hat{q}_{1}\\
	\hat{q}_{2}\\
	\vdots\\
	\hat{q}_{N}\\
	\end{pmatrix}
	\textbf{Y} =
	\begin{pmatrix}
	q_{1}\\
	q_{2}\\
	\vdots\\
	q_{N}\\
	\end{pmatrix}
	\end{align*}
	If WTP is exponentially distributed, the entries in the vectors of quantity of demand ($\hat{\textbf{Y}}$ and $\textbf{Y}$) are in logarithmic form ($\ln(\hat{q}_{i})$ and $\ln(q_{i})$). If instead WTP is Pareto distributed, the entries in $\textbf{X}$ are also in logarithmic form ($\ln(q_{i}), \ln(z_{i})$), in addition to $\hat{\textbf{Y}}$ and $\textbf{Y}$. For unknown distribution of WTP with second-order (or higher-order) of Taylor approximation, the entries in $\textbf{X}$ shall also include the corresponding higher order polynomials of $p$ and $z$ and their interaction terms. 
	
	Besides, when the Taylor expansion is used to approximate the unknown distribution of WTP at $(p_{0},z_{0})$, the estimation sample shall be restricted to a small neighborhood of $(p_{0},z_{0})$, namely $N_{\delta_{n}}(p_{0},z_{0})$ with $\underset{n\rightarrow\infty}{\lim} \delta_{n}=0 $. The number of columns of $\textbf{X}$ ($k$) also depends on the survey design. For example, if only income and competitor's price are included in the survey, that is all the other factors are implicitly controlled, the number of columns of $z$ in $\textbf{X}$ is two. 
	
	The OLS estimator is $\hat{\beta}=(\textbf{X}'\textbf{X})^{-1}\textbf{X}'\hat{\textbf{Y}}$. Since $\underset{n\rightarrow\infty}{\lim}(\hat{\textbf{Y}}-\textbf{Y})=\textbf{0}$, by construction, $\textbf{X}$ is asymptotically orthogonal to the error term ($\hat{\textbf{Y}} - \textbf{Y}$). Hence, the OLS estimator is consistent. In case of small sample, the asymptotic orthogonality is not guaranteed. Using excluded instruments ($\tilde{\textbf{Z}}$), one can use the generalized method of moments (GMM) estimator to consistently estimate $\beta$, namely $\hat{\beta}_{GMM}=(\textbf{X}' \tilde{\textbf{Z}} \tilde{\textbf{Z}}' \textbf{X})^{-1} \textbf{X}' \tilde{\textbf{Z}} \tilde{\textbf{Z}}' \hat{\textbf{Y}}$.  
	
	\begin{theorem}
		$\hat{\beta}$ is a consistent estimator of $\beta$. 
	\end{theorem}
	\textit{Proof}: 
	$
	\underset{n\rightarrow\infty}{\plim}\hat{\beta} = \underset{n\rightarrow\infty}{\plim}(\textbf{X}'\textbf{X})^{-1}\textbf{X}'\hat{\textbf{Y}}=
	\underset{n\rightarrow\infty}{\plim}(\textbf{X}'\textbf{X})^{-1}\textbf{X}'\textbf{Y}=\beta$
	, where the second equality is due to the Glivenko-Cantelli Theorem and the third equality is due to the consistency of the OLS estimator. 
	
	$\eolqed$
	
	\subsection{Miss-reporting of WTP}
	In Step 1 of Section 3.1, one can use different methods to obtain data of WTP and $Z$, which not surprisingly perform differently in revealing consumers' WTP truthfully. In the Vickrey auction, due to truthful bidding being the weakly dominant strategy, the biddings reveal participants' WTP truthfully. In the CVM, another popular method, one uses survey to solicit consumers' WTP. In the survey, a number of biases can exist in soliciting consumers' WTP. First, there can be a strategic bias. That is, the interviewees think they can obtain the benefit without telling the truth, or they are re-sell the goods. In such cases, they will not reveal the true WTP. Second, there can be design biases. The starting point can exert an anchor point effect as the interviewees take a hint of the value of the goods from the first dollar value that they are asked. To avoid this anchor point effect, researchers can randomize the first dollar value that they ask. Information is also important for the survey. Consumers need to be provided with sufficient information in order for them to report WTP truly. As such, information can distort the reporting of WTP, particularly if consumers are not familiar with the goods. Lastly, it shall be noted that the survey is conducted in a hypothetical situation, and consumers' response may be different from their actual purchase decisions. To address this issue, researchers can cross-check the survey data with consumers' actual purchase, if such data are available.
	
	Despite consumers may not report their true WTP, it turns out that the estimation procedure outlined in Section 3.1 is robust to small mis-reporting of WTP. More specifically, let $\widetilde{W}$ be the reported WTP and $\Delta$ represent the reporting error, namely $\widetilde{W}=W+\Delta$. We can derive the relationship between $F_{W|Z}(p)$ and $F_{\widetilde{W}|Z}(p)$, as follows:
	\begin{align*}
	F_{\widetilde{W}|Z}(p) &=\Pr (W\leq p-\Delta|Z) \\
	&=\Pr (\Delta \geq 0)\times \Pr (W\leq p -\Delta|Z, \Delta \geq 0) +\\
	&\Pr (\Delta < 0)\times \Pr (W\leq p -\Delta|Z, \Delta < 0) \\
	&=F_{W|Z}(p)+ \xi
	\end{align*}  
	where the second equality is due to the law of total probability, and $\xi \equiv [- \Pr (\Delta \geq 0)\times \Pr(p-\Delta <W\leq p|Z,\Delta \geq 0) + \Pr (\Delta < 0)\times Pr(p<W\leq p-\Delta|Z,\Delta < 0)]$. Note if $\xi$ is a constant or small, then the estimation procedure will work well. 
	
	\begin{theorem}
		If $W|Z,\Delta$ is uniformly distributed, the estimation procedure is robust to arbitrary mis-reporting of WTP in identifying the coefficient of price.
	\end{theorem} 
	\textit{Proof}: 
	Since $W|Z,\Delta$ is uniformly distributed, $\Pr(p-\Delta <W\leq p|Z,\Delta \geq 0)$ and $Pr(p<W\leq p-\Delta|Z,\Delta < 0)$ do not depend on $p$. Therefore, $\xi$ does not depend on $p$. Then conditional on $Z$:
	\begin{align*}
	q=[1-F_{W|Z}(p)]L=[1-F_{\widetilde{W}|Z}(p)]L-\xi L 
	\end{align*}  
	where $\xi L$ is a constant. Hence, the estimation procedure is robust to arbitrary mis-reporting of WTP.
	
	$\eolqed$
	
	Proposition 2 suggests that if the underlying true WTP is uniformly distributed, the reporting errors do not affect the estimation of the coefficient of price. Note as $\xi L$ can depend on $Z$, the estimation of coefficients of $Z$ is likely to be affected. In the following, Proposition 3 shows that if the reporting error is small and the PDF of $W$ (true WTP) is flat, the estimation procedure is robust.
	
	\begin{theorem}
		Conditional on $Z$, let $\Delta$ be independent of $W$ and randomly distributed over the support $[-\overline{\delta},\overline{\delta}]$ with $\overline{\delta}$ small and $\Pr (\Delta \geq 0) = \Pr (\Delta < 0)$, and $W|Z$ have a PDF that is flat. Then, the estimation procedure is robust in identifying the coefficient of price.
	\end{theorem} 
	\textit{Proof}: Using $\Pr (\Delta \geq 0) = \Pr (\Delta < 0)$, the difference between the CDFs of true and reported WTP can be written as follows:
	\begin{align*}
	\xi &=\frac{1}{2} \times [\int_{p}^{p+|\Delta|}f_{W|Z}(w)dw - \int_{p-|\Delta|}^{p}f_{W|Z}(w)dw] \\
	&= \frac{1}{2} \times \int_{p}^{p+|\Delta|}[f_{W|Z}(w)-f_{W|Z}(w-|\Delta|)]dw \\
	&= \frac{1}{2} \times \int_{p}^{p+|\Delta|}f'_{W|Z}(\zeta)|\Delta|dw \\
	&\leq \frac{1}{2} \times \sup f'_{W|Z} \times \overline{\delta}^2
	\end{align*} 
	where the second equality is obtained by change of variables, and the third equality is due to the mean value theorem, where $\zeta \in (w-|\Delta|,w)$. Therefore, given that the PDF, $f_{W|Z}$, is flat and $\delta$ is small, $\xi$ is bounded by a term that is close to zero. Hence, the estimation procedure is robust to small mis-reporting of WTP.
	
	$\eolqed$
	
	In the proof of Proposition 3, we can observe that whether the mis-reporting of WTP will affect estimation of the demand function critically depends on the shape (flatness) of the PDF of true WTP. If the PDF is relatively flat, such as the uniform PDF, the impact on identification is negligible. With reporting errors, the local OLS estimator (OLS estimation of a Taylor approximated demand function at a small neighborhood of the expansion point) is more desirable in that in the small neighborhood, the PDF of true WTP is more or less flat.  
	
	\section{Monte Carlo Simulations}
	In this section, I report the Monte Carlo simulation results, which illustrate the effectiveness of the estimation procedure proposed in Section 3.1. I look at four distributions of the true WTP, namely the uniform, exponential, Pareto, and log-normal distributions. For the uniform distribution, a linear demand function is estimated, while for the exponential and Pareto distributions, semi-logarithmic and logarithmic linear demand functions are estimated respectively. For the log-normal distribution, I estimate a logarithmic linear demand function, which is a first-order Taylor approximation of the underlying true demand function. In addition, I also estimate the demand functions when the reported WTP contains small errors. In general, the simulation results suggest that the estimation procedure works well in identifying the coefficient of price. 
	
	The data generating processes are as follows\footnote{I assume WTP does not depend on $Z$.}: (1) the number of consumers in the market is set to 1000; (2) 200 observations of WTP are drawn from the distributions of $W \sim U[0,50]$, $W\sim Exp(1)$, $W\sim Pareto(1,2)$, and $W\sim Lognormal(0,1)$ respectively, and 200 small reporting errors are drawn from a uniform distribution, $U[-0.5,0.5]$; Note that it implies the following associated true demand functions: 
	\begin{align*}
	q&=1000-20p \\
	\log(q)&=\log(1000)-p \\
	\log(q)&=\log(1000)-2p \\
	q&=[\frac{1}{2} - \frac{1}{2}erf[\frac{\log(p)}{\sqrt{2}}]]\times 1000
	\end{align*} 
	(3) the empirical CDF, $F_{\widetilde{W}|Z}(w)$, and the quantity of demand are estimated accordingly, and to account for possible mis-reporting, the error draws are added to the WTP draws before the estimation of empirical CDF and quantity of demand; (4)the constructed data are then used in regressions to estimate the coefficients of price.  
	
	Table 1 reports the estimation results with the simulated data (without reporting errors)\footnote{Tables in this paper are created in R using stargazer v.5.2.2 by Marek Hlavac.}, and Figure 1 plots the corresponding scatter plots and fitted lines. It can be observed that the estimation procedure largely works well. For the uniform distribution, the estimated coefficient of price is -21.234, while the true coefficient is -20. The estimation error of 6.17 percent is small. For the exponential distribution, the price coefficient is estimated to be -1.097, which is close to the true value of -1. The coefficient of price for the Pareto distribution is estimated to be -1.781, with the true value being -2. The estimation error is 10.95 percent. As for the log-normal distribution, since I use a first-order Taylor approximation of the underlying demand function, it is not surprising to observe that the estimation is less accurate (see the right bottom panel of Figure 1). To estimate the approximated demand function more accurately, one needs to restrict the sample to a small neighborhood of the expansion point. Here we estimate the approximated demand function using all the available data.  
	
	\begin{table}[!htbp] \centering 
		\caption{Simulation Results}  
		\label{} 
		\small 
		\begin{tabular}{@{\extracolsep{-15pt}}lD{.}{.}{-3} D{.}{.}{-3} D{.}{.}{-3} D{.}{.}{-3} } 
			\\[-1.8ex]\hline 
			\hline \\[-1.8ex] 
			& \multicolumn{4}{c}{\textit{Dependent variable:}} \\ 
			\cline{2-5} 
			\\[-1.8ex] & \multicolumn{1}{c}{$q$} & \multicolumn{3}{c}{$\log(q)$} \\ 
			\\[-1.8ex] & \multicolumn{1}{c}{Uniform} & \multicolumn{1}{c}{Exponential} & \multicolumn{1}{c}{Pareto} & \multicolumn{1}{c}{Log-normal}\\ 
			\hline \\[-1.8ex] 
			$price$ & -21.234^{***} & -1.097^{***} &  &  \\ 
			& (0.114) & (0.009) &  &  \\ 
			& & & & \\ 
			$\log(price)$ &  &  & -1.781^{***} & -0.885^{***} \\ 
			&  &  & (0.013) & (0.027) \\ 
			& & & & \\ 
			Constant & 1,010.158^{***} & 7.026^{***} & 6.841^{***} & 5.915^{***} \\ 
			& (3.150) & (0.013) & (0.010) & (0.027) \\ 
			& & & & \\ 
			\hline \\[-1.8ex] 
			Observations & \multicolumn{1}{c}{200} & \multicolumn{1}{c}{199} & \multicolumn{1}{c}{199} & \multicolumn{1}{c}{199} \\ 
			R$^{2}$ & \multicolumn{1}{c}{0.994} & \multicolumn{1}{c}{0.986} & \multicolumn{1}{c}{0.990} & \multicolumn{1}{c}{0.842} \\ 
			Adjusted R$^{2}$ & \multicolumn{1}{c}{0.994} & \multicolumn{1}{c}{0.985} & \multicolumn{1}{c}{0.990} & \multicolumn{1}{c}{0.842} \\ 
			\hline 
			\hline \\[-1.8ex] 
			\textit{Note: }  & \multicolumn{4}{r}{$^{*}$p$<$0.1; $^{**}$p$<$0.05; $^{***}$p$<$0.01. } \\ 
		\end{tabular} 
	\end{table}
	
	\begin{figure}
		\caption{Simulation Results}
		\label{figure1}
		\includegraphics[width=0.8\textwidth]{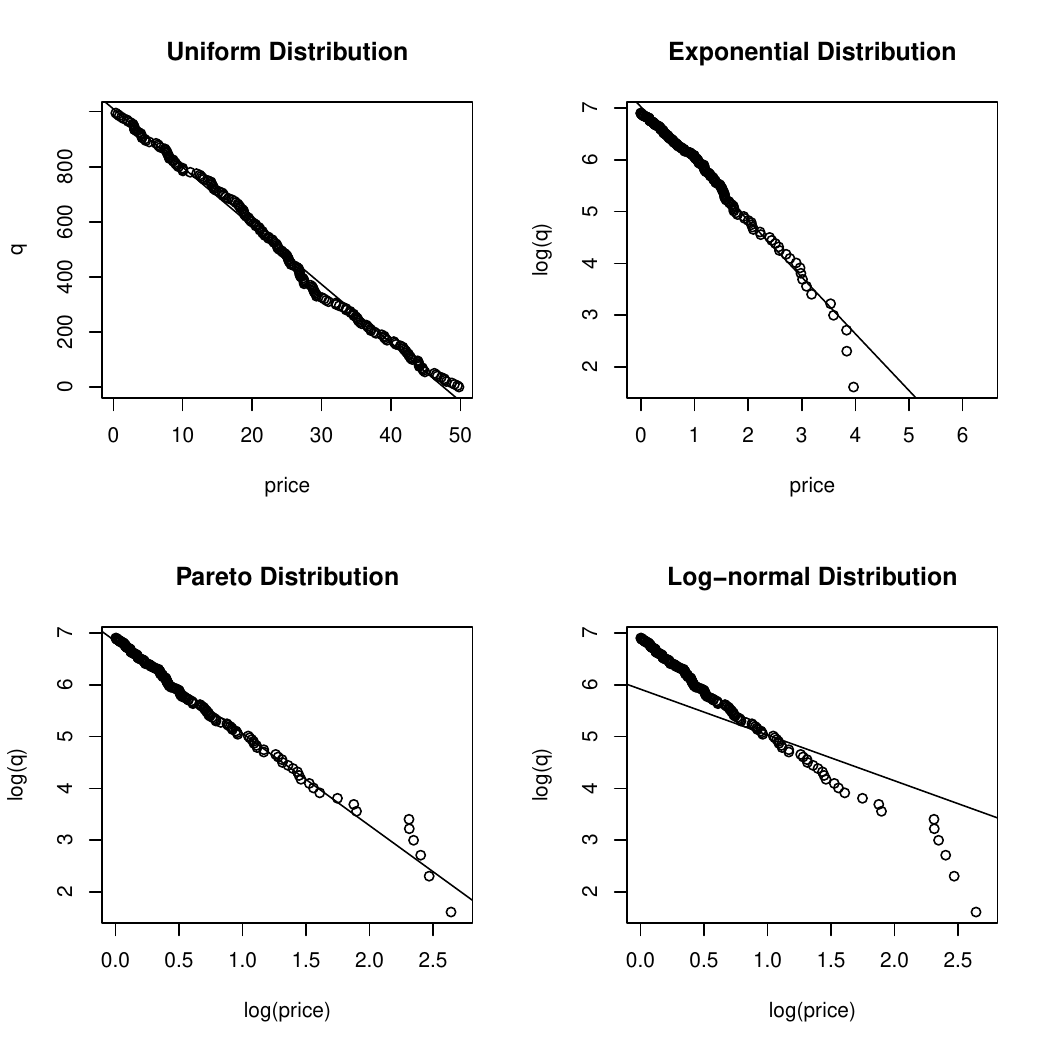}
		\centering
	\end{figure}
	
	Table 2 presents the simulation results when the reported WTP contains small errors. For the uniform distribution, the coefficient of price is estimated to be -21.098, which is close to the true value of -20 (with 5.49 percent estimation error). In addition, compared with the estimate of Table 1, the point estimates appear not to be affected by the reporting errors, which is consistent with Proposition 1. When WTP is exponentially distributed (with small uniform distributed errors), the estimation also identifies the coefficient of price accurately, with the point estimate being -0.929 (true value being -1).
	
	In contrast, the point estimates, with Pareto distributed WTP and small uniformly distributed reporting errors, are now less accurate, with the estimated coefficient of price being -1.463. The estimation error is 26.85 percent, higher than that of Table 1 (10.95 percent estimation error). Hence, it appears that the Pareto distributed WTP is less robust to mis-reporting. Similarly, for the log-normal distribution of WTP, the point estimates also exhibit variations, compared with those without reporting errors. As in the proof of Proposition 3, the impact of reporting errors critically depends on the flatness of the PDF of true WTP. Both Pareto and log-normal distributions are less flat than the uniform and exponential distributions.  
	
	\begin{table}[!htbp] \centering 
		\caption{Simulation Results with Small Reporting Errors} 
		\label{} 
		\small 
		\begin{tabular}{@{\extracolsep{-15pt}}lD{.}{.}{-3} D{.}{.}{-3} D{.}{.}{-3} D{.}{.}{-3} } 
			\\[-1.8ex]\hline 
			\hline \\[-1.8ex] 
			& \multicolumn{4}{c}{\textit{Dependent variable:}} \\ 
			\cline{2-5} 
			\\[-1.8ex] & \multicolumn{1}{c}{$q$} & \multicolumn{3}{c}{$\log(q)$} \\ 
			\\[-1.8ex] & \multicolumn{1}{c}{Uniform} & \multicolumn{1}{c}{Exponential} & \multicolumn{1}{c}{Pareto} & \multicolumn{1}{c}{Log-normal}\\  
			\hline \\[-1.8ex] 
			$price$ & -21.098^{***} & -0.929^{***} &  &  \\ 
			& (0.123) & (0.026) &  &  \\ 
			& & & & \\ 
			$\log(price)$ &  &  & -1.463^{***} & -0.663^{***} \\ 
			&  &  & (0.044) & (0.037) \\ 
			& & & & \\ 
			Constant & 1,007.719^{***} & 6.885^{***} & 6.632^{***} & 5.838^{***} \\ 
			& (3.393) & (0.039) & (0.034) & (0.043) \\ 
			& & & & \\ 
			\hline \\[-1.8ex] 
			Observations & \multicolumn{1}{c}{200} & \multicolumn{1}{c}{174} & \multicolumn{1}{c}{200} & \multicolumn{1}{c}{189} \\ 
			R$^{2}$ & \multicolumn{1}{c}{0.993} & \multicolumn{1}{c}{0.882} & \multicolumn{1}{c}{0.850} & \multicolumn{1}{c}{0.631} \\ 
			Adjusted R$^{2}$ & \multicolumn{1}{c}{0.993} & \multicolumn{1}{c}{0.882} & \multicolumn{1}{c}{0.850} & \multicolumn{1}{c}{0.629} \\ 
			\hline 
			\hline \\[-1.8ex] 
			\textit{Note:}  & \multicolumn{4}{r}{$^{*}$p$<$0.1; $^{**}$p$<$0.05; $^{***}$p$<$0.01} \\ 
		\end{tabular} 
	\end{table} 
	
	\begin{figure}
		\caption{Simulation Results with Small Reporting Errors}
		\includegraphics[width=0.8\textwidth]{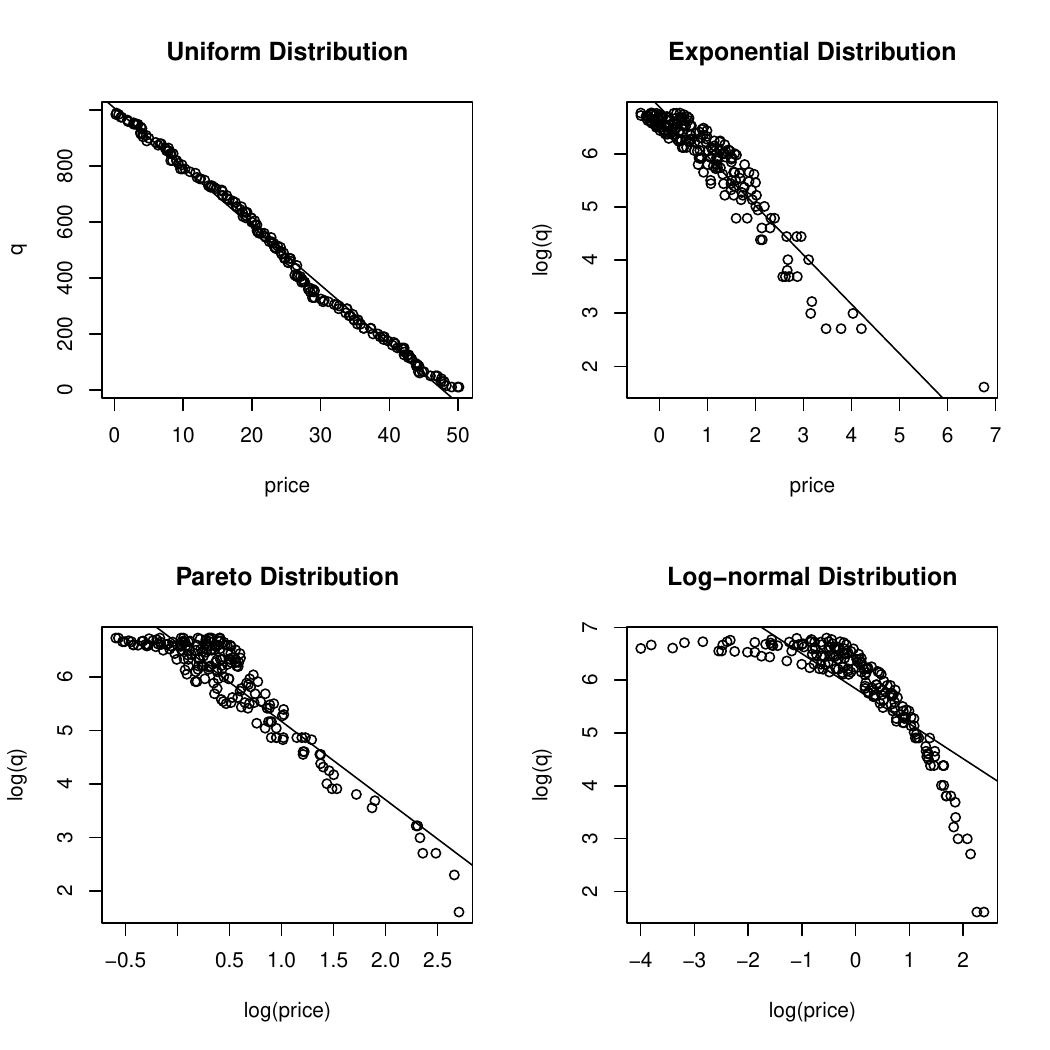}
		\centering
	\end{figure}
	
	\section{Demand for Coffee}
	Coffee is a goods with which consumers are familiar, and as such serves as a good example for illustrating estimation of demand functions, despite it is not a new product. Utilizing the estimation procedure in Section 3.1 and CVM, I estimate a demand for coffee (medium-sized long black) by the staff of James Cook University\footnote{James Cook University Ethics Approval Number H7972.}.
	
	The survey questionnaire is designed to solicit the WTP of coffee consumers at James Cook University. First, the survey starts with asking whether the interviewees are willing to pay $\$ x$  for a cup of medium-sized long black, where $x \in \{4, 8, 16, 32\}$ is randomly presented to interviewees to control for the possible anchor point effect. If an interviewee answers yes, she/he is then asked whether she/he is willing to pay $\$2x$. If again the interview answers yes, she/he is asked to report a WTP in the range of $[\$2x, \$50]$. If the interviewee is not willing to pay $\$2x$, she/he is asked to report a WTP in the range of $[\$x, \$2x)$. If in the first question, the interviewee answers no, namely not willing to pay $\$ x$ for the coffee, she/he is asked whether she/he is willing to pay $\$\frac{1}{2} x$. If answered yes, she/he is directed to report a WTP in the rage of $[\$\frac{1}{2}x, \$x)$, and if no the range of WTP is $[\$0, \$\frac{1}{2}x)$. Figure 3 presents the flow chart. The survey was implemented using Qualtrics. The survey questionnaire is attached in the appendix.
	
	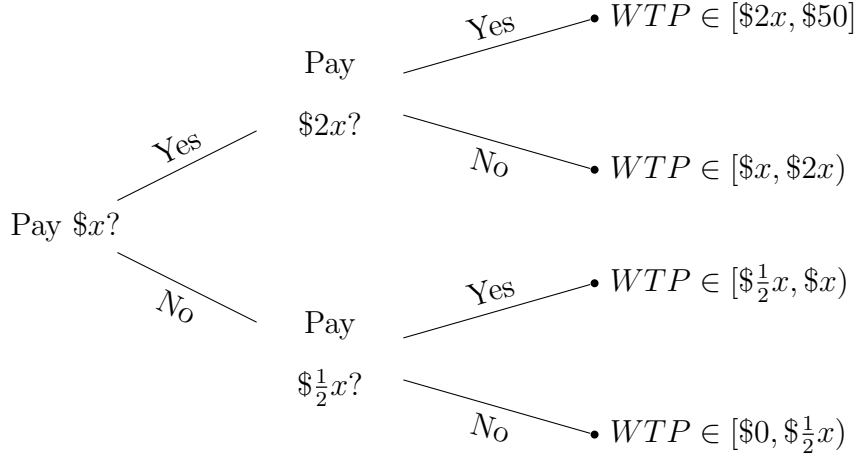
\begin{figure}
		\caption{The Flow Chart of Soliciting WTP}
		\tikzstyle{level 1}=[level distance=3.5cm, sibling distance=3.5cm]
		\tikzstyle{level 2}=[level distance=3.5cm, sibling distance=2cm]
		
		\tikzstyle{bag} = [text width=4em, text centered]
		\tikzstyle{end} = [circle, minimum width=3pt,fill, inner sep=0pt]
		
		\begin{tikzpicture}[grow=right, sloped]
		\node[bag] {Pay $\$x$?}
		child {
			node[bag] {Pay $\$\frac{1}{2}x$?}        
			child {
				node[end, label=right:
				{$WTP\in [\$0,\$\frac{1}{2}x)$}] {}
				edge from parent
				node[above] {}
				node[below] {No}
			}
			child {
				node[end, label=right:
				{$WTP\in [\$\frac{1}{2}x,\$x)$}] {}
				edge from parent
				node[above] {Yes}
				node[below]  {}
			}
			edge from parent 
			node[above] {}
			node[below]  {No}
		}
		child {
			node[bag] {Pay $\$2x$?}        
			child {
				node[end, label=right:
				{$WTP\in [\$x,\$2x)$}] {}
				edge from parent
				node[above] {}
				node[below]  {No}
			}
			child {
				node[end, label=right:
				{$WTP\in [\$2x,\$50]$}] {}
				edge from parent
				node[above] {Yes}
				node[below]  {}
			}
			edge from parent         
			node[above] {Yes}
			node[below]  {}
		};
		\end{tikzpicture}
	\end{figure}
	
	The survey design is simple and sufficient for the purpose of illustration. Not surprisingly, there are likely other factors that can affect a consumer's WTP for a cup of coffee. In this design, these other factors are implicitly controlled. However, one can easily revise the design to suit different purposes of research. For example, if a firm intends to investigate how a competitor's price affects the demand for its product, it can incorporate the competitor's price in the survey design, and solicit consumers' WTP for its product at different levels of competitor prices.
	
	\begin{figure}
		\caption{Distribution of WTP by JCU Students and Staff}
		\includegraphics[width=0.8\textwidth]{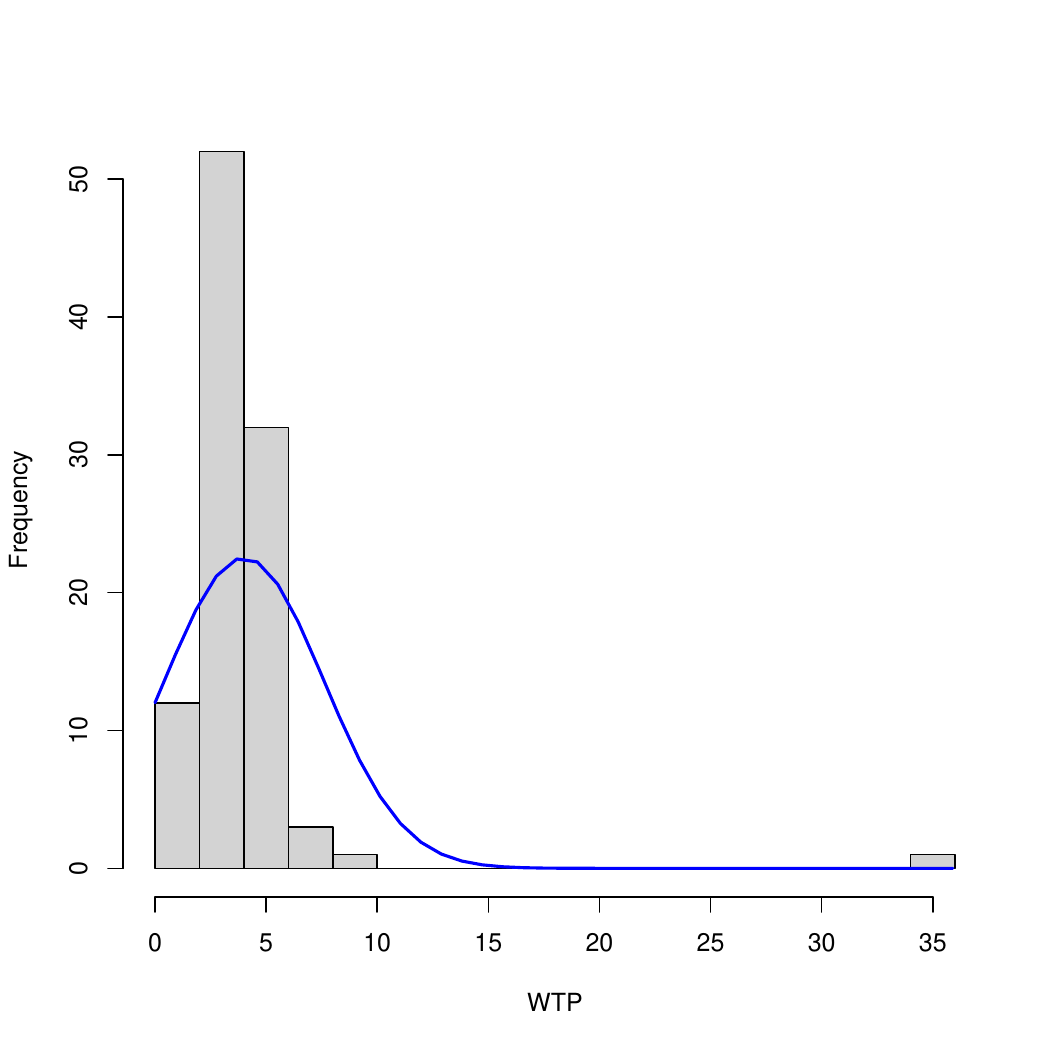}
		\centering
	\end{figure} 
	
	Figure 4 reports the distribution of WTP for a cup of coffee (medium-sized long black) at James Cook University ($n=101$). It can be observed that all interviewees, except one, are willing to pay less than \$10 for a cup of medium-sized long back, with the biggest proportion willing to pay between \$4 and \$6. In addition, the distribution appears not to resemble a normal distribution.
	
	\begin{table}[!htbp] \centering 
		\caption{Estimation Results} 
		\label{} 
		\small 
		\begin{tabular}{@{\extracolsep{-15pt}}lD{.}{.}{-3} D{.}{.}{-3} D{.}{.}{-3} } 
			\\[-1.8ex]\hline 
			\hline \\[-1.8ex] 
			& \multicolumn{3}{c}{\textit{Dependent variable:}} \\ 
			\cline{2-4} 
			\\[-1.8ex] & \multicolumn{1}{c}{q} & \multicolumn{2}{c}{log(q)} \\ 
			\\[-1.8ex] & \multicolumn{1}{c}{(1)} & \multicolumn{1}{c}{(2)} & \multicolumn{1}{c}{(3)}\\ 
			\hline \\[-1.8ex] 
			price & -157.023^{***} &  & -0.550^{***} \\ 
			& (38.581) &  & (0.041) \\ 
			& & & \\ 
			log(price) &  & -0.487^{***} &  \\ 
			&  & (0.138) &  \\ 
			& & & \\ 
			Constant & 3,029.869^{***} & 7.867^{***} & 9.410^{***} \\ 
			& (261.742) & (0.215) & (0.171) \\ 
			& & & \\ 
			\hline \\[-1.8ex] 
			Observations & \multicolumn{1}{c}{44} & \multicolumn{1}{c}{42} & \multicolumn{1}{c}{43} \\ 
			R$^{2}$ & \multicolumn{1}{c}{0.283} & \multicolumn{1}{c}{0.237} & \multicolumn{1}{c}{0.812} \\ 
			Adjusted R$^{2}$ & \multicolumn{1}{c}{0.266} & \multicolumn{1}{c}{0.218} & \multicolumn{1}{c}{0.807} \\ 
			\hline 
			\hline \\[-1.8ex] 
			\textit{Note:}  & \multicolumn{3}{r}{$^{*}$p$<$0.1; $^{**}$p$<$0.05; $^{***}$p$<$0.01} \\ 
		\end{tabular} 
	\end{table}    
	
	\begin{figure}
		\caption{Demand for Coffee}
		\includegraphics[width=0.8\textwidth]{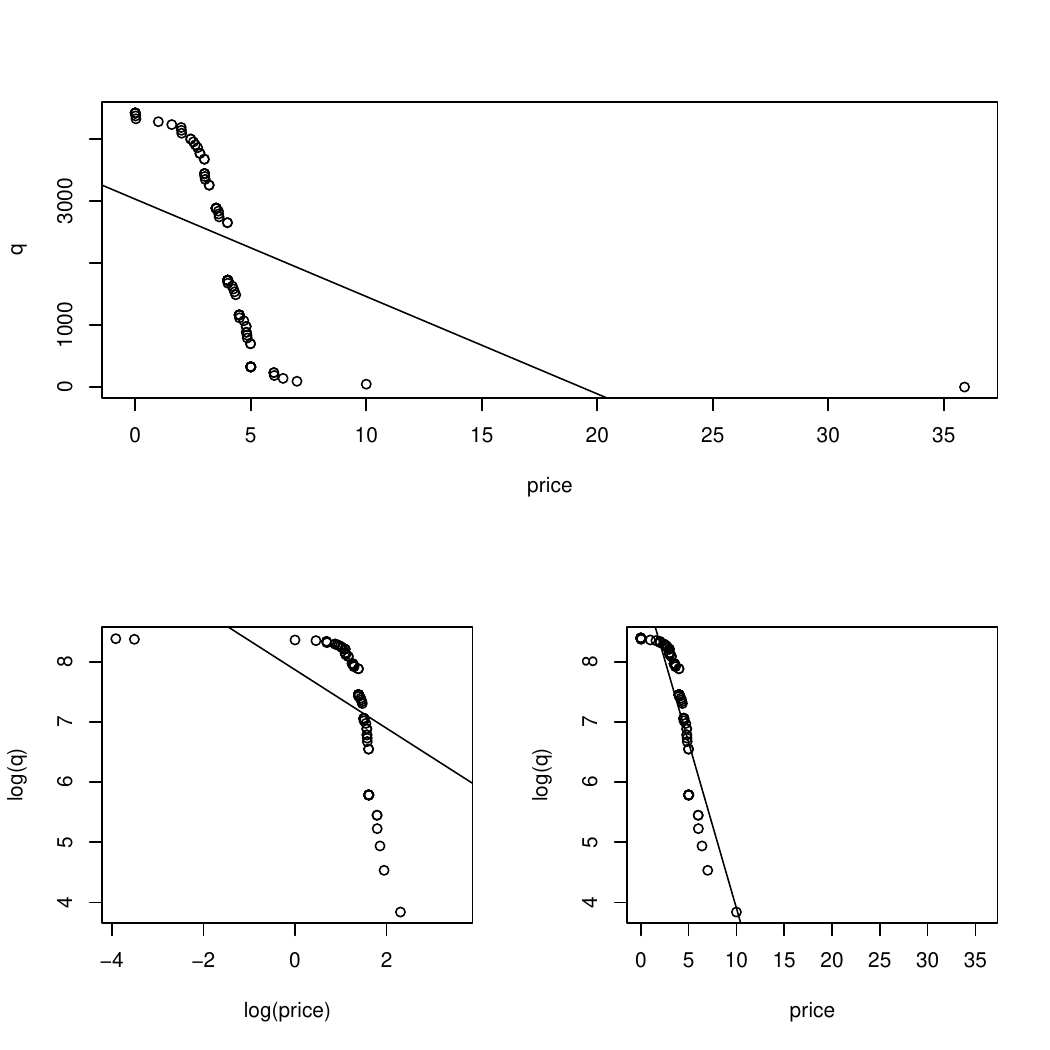}
		\centering
	\end{figure}
	
	Table 3 reports the regression results, and Figure 5 presents the scatter plots with fitted lines, where the number of consumers ($L$) is set to 4698. For the linear demand function, the coefficient of price is estimated to be -157.023, namely a one dollar increase in the price resulting in a reduction of demand by around 157 cups. For the logarithmic linear demand function, the price elasticity of demand (coefficient of $\log(price)$) is estimated to be -0.487. The inelastic demand for coffee is intuitive since for many coffee is a daily necessity. The estimation of the semi-logarithmic linear demand function finds the coefficient of price to be -0.550, namely a one dollar increase in coffee price leads to 55 percent decrease in the demand. In all estimations, the coefficient estimates are statistically significant at the one percent level. Comparing the three estimations, it appears that the semi-logarithmic linear model fits the data well in that it has the highest adjusted $R^2$.
	
	\section{Concluding Remarks}
	Taking WTP as primitive, this paper proposes a demand function for a goods that is analytically simple and can be easily used in economic modelling of more complex situations. In addition, the study also establishes an estimation procedure that is easy to be implemented, can estimate the underlying parameters consistently, and is robust to mis-reporting of WTP. The estimation procedure can be applied in a number of occasions. 
	
	One is for economics education. In a class room environment, a lecturer can quickly survey students' WTP for a goods, for example coffee. With the survey data and the help from a spreadsheet program, such as Microsoft Excel, the lecturer can demonstrate the concept of demand in an intuitive way. By construction, the Law of Demand is guaranteed to be satisfied. The lecturer can also design a project that requires students to conduct a survey to estimate the demand function for a particular goods.
	
	A second situation that the estimation procedure can be applied is for firms to estimate the demand for their new products, the focus of this paper. In a firm's marketing research on launching a new product, there are no observational data for demand estimation due to the fact that the product is new to the market. Hence, estimating the demand function through soliciting consumers' WTP is feasible and particularly attractive. 
	
	A third situation is for the non-market valuation. For example, one can estimate a demand function for reduction of CO2 emission. Knowledge of such a demand function allows policymakers to assess the welfare impact of their emission reduction plan.
	
	Given the focus on new products, the demand function is necessarily static. In addition, in many cases, consumers' purchase decisions are static, for example decision on whether to purchase a cup of coffee in a day. However there are situations where such purchase decision is more or less dynamic, particularly if the goods is durable, such as purchase of a car. The dynamic demand function, using WTP as model primitive, is a direction for future research.
	
	\newpage
	
	\printbibliography
	
	\newpage
	\appendix
	\section{Survey Questionnaire}
	\begin{figure}[H]
		\captionsetup{labelformat=empty}
		\caption{}
		\includegraphics[width=1\textwidth]{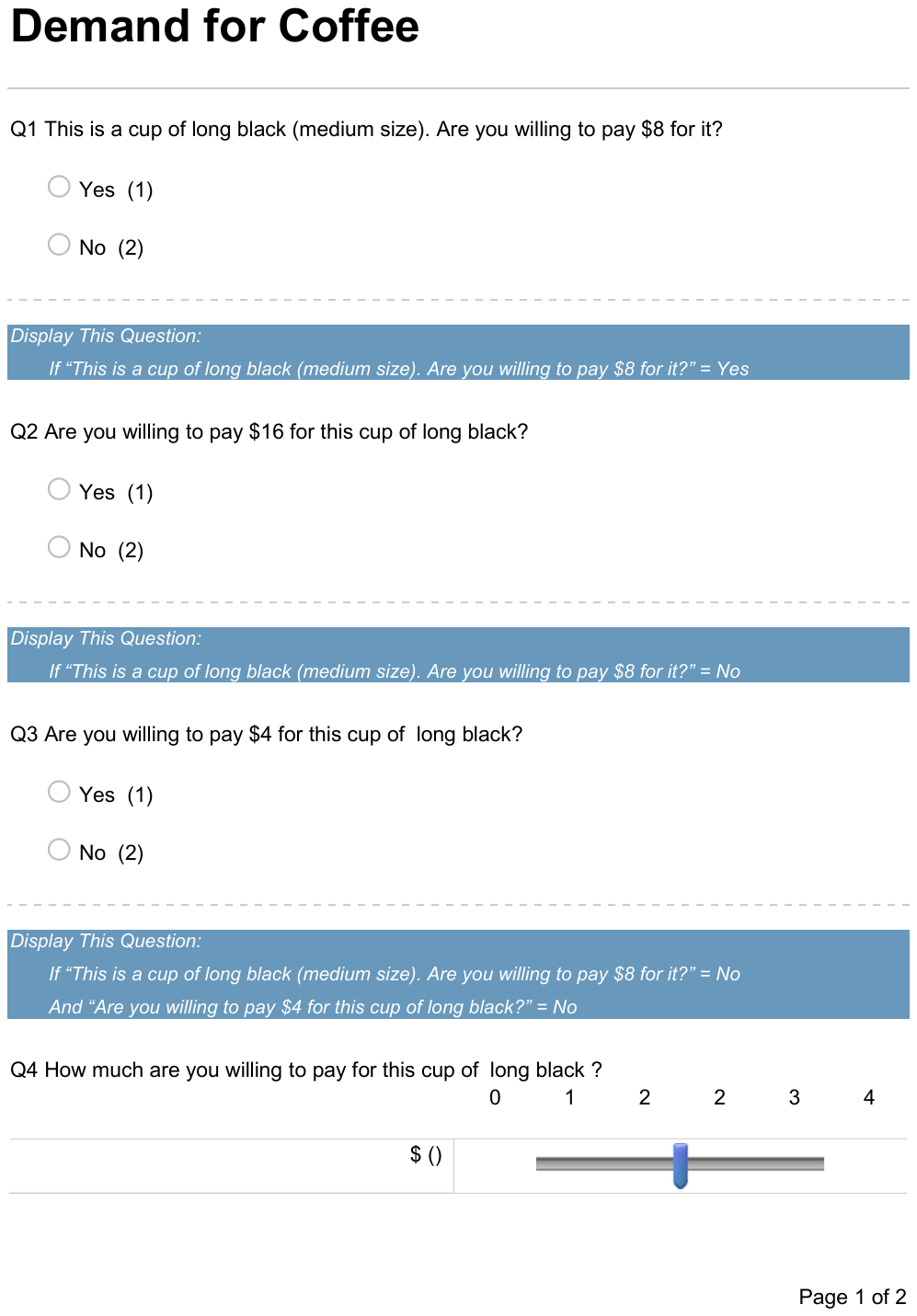}
		\centering
	\end{figure}
	\newpage
	\begin{figure}[t]
		\captionsetup{labelformat=empty}
		\caption{}
		\includegraphics[width=1\textwidth]{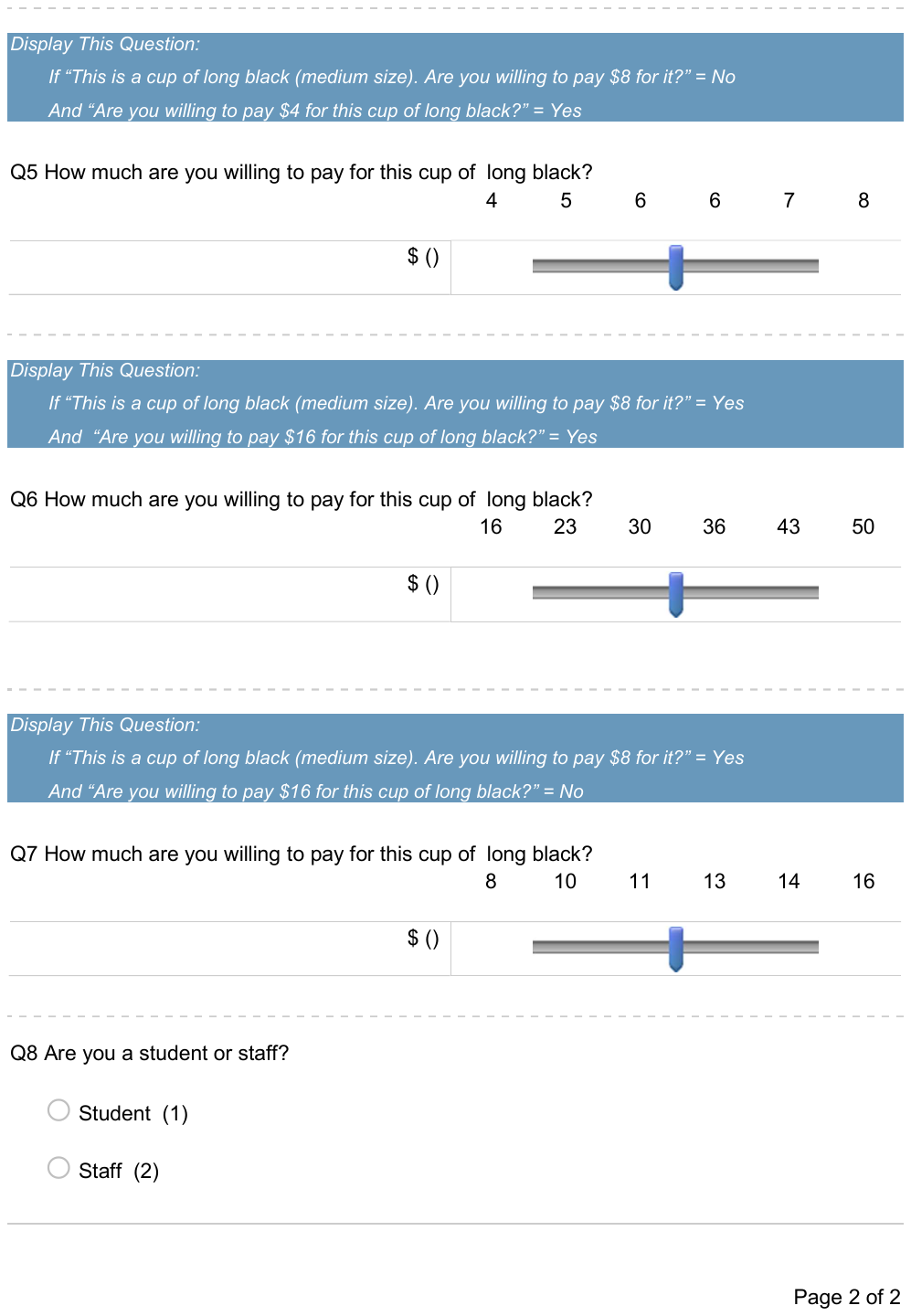}
		\centering
	\end{figure}
\end{document}